\documentclass[reprint,nofootinbib]{revtex4-2}
\usepackage{amssymb,amsmath,amsthm}
\usepackage{bm,bbm}
\usepackage{graphicx}
\usepackage{overpic}
\usepackage[bookmarks=false]{hyperref}
\hypersetup{colorlinks=true,citecolor=blue,linkcolor=red,%
urlcolor=blue,pdfstartview=FitH,bookmarksopen=true}

\usepackage{array}
\usepackage{adjustbox}
\usepackage[T1]{fontenc}
\usepackage[osf,sc]{mathpazo}
\usepackage{dsfont}
\usepackage{booktabs}
\usepackage{algorithm}
\usepackage{algorithmicx}
\usepackage{algpseudocode}

\usepackage{silence}
\newcommand{\ket}[1]{\mbox{$ | #1 \rangle $}}
\newcommand{\bra}[1]{\mbox{$ \langle #1 | $}}

\newcommand{\tr}{\mathrm{tr}}
\newcommand{\Tr}{\mathrm{Tr}}

\newcommand{\cM}{\mathcal{M}}

\newcommand{\cT}{\mathcal{T}}

\newcommand{\cC}{\mathbb{C}}
\newcommand{\cE}{\mathcal{E}}

\newcommand{\cI}{\mathcal{I}}

\newcommand{\cA}{\mathcal{A}}
\newcommand{\cD}{\mathcal{D}}
\newcommand{\cX}{\mathcal{X}}
\newcommand{\cV}{\mathcal{V}}

\newcommand{\bR}{\mathbb{R}}
\newcommand{\bE}{\mathbb{E}}

\newcommand{\bC}{\mathbb{C}}

\newcommand{\bu}{{\mathbf u}}

\newcommand{\by}{{\bf y}}
\newcommand{\bv}{{\bf v}}
\newcommand{\bz}{{\bf z}}
\newcommand{\bc}{{\bm{\chi}}}

\newcommand{\minover}[1][]{\underset{#1}{\text{min}}}
\newcommand{\subto}{\text{s.t.}}
\newcommand{\invplus}{\phantom{+}}
\newtheoremstyle{note}
  {\topsep/2}              	
  {\topsep/2}            	
  {}                        
  {\parindent}             	
  {\itshape}                
  {.---}                    
  {0pt}                     
  {\thmname{#1}\thmnumber{ \itshape#2}\thmnote{ (#3)}} 

\newtheorem{theorem}{Theorem}

\newtheorem{proposition}[theorem]{Proposition}

\theoremstyle{definition}
\newtheorem{definition}{Definition}

\theoremstyle{remark}
\newtheorem{remark}{Remark}

\begin{document}
\title{Simultaneous reconstruction of quantum process and noise via corrupted sensing}

\author{Mengru Ma}
\author{Jiangwei Shang}
\email{jiangwei.shang@bit.edu.cn}
\affiliation{Key Laboratory of Advanced Optoelectronic Quantum Architecture and Measurement (MOE), School of Physics, Beijing Institute of Technology, Beijing 100081, China}

\date{\today}

\begin{abstract}
Quantum processes, including quantum gates and channels, are integral to various quantum information tasks, making the efficient characterization of these processes and their underlying noise critically important. Here, we propose a framework for quantum process tomography in the presence of corrupted noise that is able to simultaneously reconstruct the process and corrupted noise. Firstly, within the Choi-state representation, we derive the corresponding generalized restricted isometry property and demonstrate the simultaneous reconstruction of various quantum gates under sparse noise. Moreover, in comparison with the Choi-state scheme, the process-matrix representation is employed to simultaneously reconstruct sparse noise and a broader range of target quantum gates. Our results demonstrate that significant reduction in experimental configurations is achievable
even under corrupted noise.
\end{abstract}
\maketitle

\section{Introduction}
The characterization of quantum gates and channels are usually carried out by the fundamental tool of quantum process tomography (QPT), which plays a central role in quantum information processing and quantum error correction~\cite{nielsen2010quantum, PhysRevLett.90.193601, PhysRevLett.93.080502, PhysRevLett.97.170501, science.1162086, PhysRevA.77.032322,gebhart2023learning}. However, standard QPT scales poorly: reconstructing a general $d$-dimensional process requires $O(d^4)$ measurement configurations and parameters, which becomes prohibitive even for moderate system sizes. Consequently, to mitigate the problem of exponential scaling of QPT resources, a number of
alternative methods have been developed for efficient and
selective estimation of quantum processes~\cite{science.1145699, PhysRevLett.100.190403,Branderhorst_2009}, such as randomized benchmarking~\cite{PhysRevA.77.012307, PhysRevLett.109.080505,PhysRevLett.121.170502} and Monte Carlo process certification~\cite{PhysRevLett.106.230501, PhysRevLett.107.210404,PhysRevLett.108.260506}. These protocols estimate the average gate fidelity rather than reconstruct the process matrices. In particular, the protocols inspired by compressed sensing (CS) techniques~\cite{PhysRevLett.105.150401,riofrio2017experimental} have been applied to QPT, which can dramatically reduce the required number of configurations.

If the process matrix is (nearly) sparse in a known basis (e.g., the Pauli bases), it can be efficiently estimated by minimizing its $\ell_1$-norm~\cite{PhysRevLett.106.100401}. While the process $\cE$ was recovered in Ref.~\cite{Flammia2012} by learning the expectation values of its Jamiołkowski state $\rho_\cE$ without using an ancilla. It is claimed that ${O(rd^2\log d)}$ settings were able to characterize $\cE$, where $r$ is the rank of $\rho_{\cE}$ (note that $\cE$ has Kraus rank $r$). Reference~\cite{PhysRevA.90.012110} studied quantum process tomography given the prior information that the process is a unitary or close to a
unitary map, reducing the required data needed for high-fidelity estimation. Moreover, Ref.~\cite{PhysRevB.90.144504} applied compressed sensing techniques to superconducting qubit systems, showing that the number of experimental configurations for a three-qubit Toffoli gate can be reduced by a factor of $\mathord{\sim}\,40$ as compared to standard QPT.
These advances collectively highlight the potential of structured QPT methods in order to enable scalable quantum characterization.

Beyond the latent structure of the process $\cE$ such as sparsity or low Kraus rank~\cite{Kliesch2019}, alternatively, other structural assumptions have also been explored to improve the reconstruction performance. Matrix product states/operators tomography,
exploiting low entanglement structure of the states, enables compact and efficient representations~\cite{cramer2010efficient,PhysRevLett.111.020401,lanyon2017efficient,RevModPhys.93.045003}. For example, tomography schemes that use only a linear number of experimental operations, together with post-processing that scales polynomially with the system size, are particularly effective for states well approximated by matrix product states~\cite{cramer2010efficient}. Similarly, process tomography techniques based on a locally-purified density operator take advantage of
the small bond dimension to characterize the quantum hardware operating circuits of
sufficiently low depth~\cite{torlai2023quantum}. Meanwhile, Bayesian and machine learning–inspired approaches incorporate geometric priors, such as manifold constraints, to learn quantum processes directly from noisy data. These techniques offer flexible, data-driven alternatives that is able to adapt to varying experimental conditions.

Despite these developments, practical implementations of QPT often overlook the impact of corrupted noise in the measurement outcomes. In realistic experiments, occasional measurement errors, detector imperfections, or calibration drifts introduce sparse but potentially corruptions that are not well captured by conventional noise models. Recently, the corrupted sensing quantum state tomography (CSQST) framework was proposed to enable the simultaneous recovery of quantum states and sparse corruption~\cite{Ma_2025}. This naturally motivates extending the simultaneous tomography of quantum states to the characterization of quantum processes, enabling the reconstruction of both the process and the underlying noise.

This work is organized as follows. In Sec.~\ref{Sec2}, we present the corrupted sensing QPT framework via
Choi state and derive the associated generalized restricted isometry property. While Sec.~\ref{Sec:process} introduces the process-matrix scheme for a comparative numerical analysis. Finally, Sec.~\ref{Sec:conclusion} concludes with a summary and outlook for future work.

\section{Corrupted sensing quantum process tomography via Choi State}\label{Sec2}

We first review the restricted isometry property (RIP), from which the generalized RIP, termed GRIP, will be introduced.
The linear map $\cA(\cdot)$ is said to have RIP with a restricted isometric constant ${\delta\in [0,1)}$ if it has the
following property~\cite{SIAMReview}:
\begin{equation}\label{RIP}
  (1-\delta)\|X\|_F^2\leq\|\cA (X)\|_2^2\leq(1+\delta)\|X\|_F^2\,.
\end{equation}
Note that for a vector $\mathbf{x}$, ${\|\mathbf{x}\|_2 = (\sum_i |x_i|^2)^{1/2}}$ denotes the $\ell_2$-norm,
${\|\mathbf{x}\|_1 = \sum_i |x_i|}$ is the $\ell_1$-norm, and $\|\mathbf{x}\|_0$ gives the number of non-zero components of $\mathbf{x}$. 
For a matrix $X$, ${\|X\|_F = (\sum_i \sigma_i(X)^2)^{1/2}}$ represents its Frobenius norm, with $\sigma_i(X)$ being the singular values of $X$. The rank of ${X\in \bC^{d\times d}}$ is at most $r$. Alternatively, the restricted
isometry constant of $\cA$ can be written as
\begin{equation}\label{supRIP}
  \delta_r=\sup_{X\in \cD_r}\left|\|\cA (X)\|_2^2-\|X\|_F^2 \right|,
\end{equation}
where $\cD_r=\{X\in\bC^{d\times d}\!:\text{rank}(X)\leq r,\|X\|_F^2\leq 1\}$, and $\delta_r$ is the smallest value of $\delta$.

\begin{definition}
For the extended matrix ${\cM=[\cA,I]\in \bC^{m\times (d+m)}}$, it has the GRIP with constant $\delta_{r, s}$ if $\delta_{r, s}$ is the smallest value of $\delta$ such that
\begin{equation}\label{GRIP}
\begin{split}
  (1-\delta)(\|X\|_F^2+\|\bv\|_2^2)
  &\leq \biggl\|\mathcal{M}
    \begin{bmatrix} X \\ \bv \end{bmatrix}\biggr\|_2^2 \\
  &\le (1+\delta)(\|X\|_F^2+\|\bv\|_2^2)
\end{split}
\end{equation}
holds for the matrix ${X\in \bC^{d\times d}}$ with rank at most $r$ and ${\bv\in \bR^{m}}$ with \emph{sparsity} ${\|\bv\|_0\leq s}$. 
\end{definition}

Note that the sensing
matrix $\mathcal{A}$ satisfying the RIP does not necessarily imply that the associated matrix $\mathcal{M}$ would satisfy the GRIP~\cite{Zhang2018}.

\subsection{GRIP for corrupted sensing QPT}
\begin{figure*}[!t]
\includegraphics[width=0.7\textwidth]{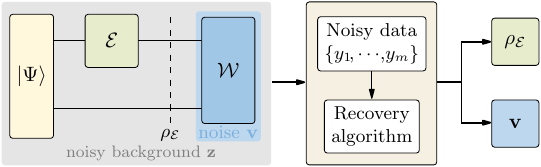}
    \caption{Schematic procedure of the corrupted sensing QPT based on the Choi-state representation. The maximally entangled state $\ket{\Psi}$ is prepared and sent into the quantum process $\mathcal{E}$. The output Choi state $\rho_{\mathcal{E}}$ is then measured to obtain the noisy data $\by = \{y_1, \ldots, y_m\}$. We randomly select $m$ Pauli observables $\{\mathcal{W}_1, \ldots, \mathcal{W}_m\}$, with $m \leq d^4$. By using the recovery algorithm, the Choi state $\rho_{\cE}$ and the measurement noise $\bv$ can be reconstructed simultaneously, and $\bz$ denotes any other potential noise.}
    \label{flow:Choi}
\end{figure*}

The GRIP condition for corrupted sensing quantum state tomography is given in Appendix~\ref{app:RGstate}.
Here, considering the Choi-Jamiołkowski isomorphism~\cite{JAMIOLKOWSKI1972275,CHOI1975285}, the process ${\mathcal{E}\!: \mathbb{C}^{d \times d} \to \mathbb{C}^{d \times d}}$ is completely and uniquely
characterized by the state
\begin{equation}\label{eq:ChoiState}
  \rho_{\cE}=(\cE\otimes \cI)(\ket{\Psi}\bra{\Psi})\,,
\end{equation}
where ${\ket{\Psi}=\frac{1}{\sqrt{d}}\sum_{j=0}^{d-1}\ket{j}\otimes \ket{j}}$. The state ${\rho_{\mathcal{E}} \in \mathbb{C}^{d^2 \times d^2}}$ is assumed to satisfy ${\Tr(\rho_{\mathcal{E}}) = 1, \rho_{\mathcal{E}} \succeq 0}$, and has rank at most $r$.

Let $P^S$ and $P^A$ be the $n$-qubit Pauli operators acting on the system and ancilla, respectively. And
$P^S, P^A \in \mathbb{C}^{d \times d}$, where ${d = 2^n}$.
We define the Pauli observable as
\begin{equation}
  \mathcal{W}= P^S \otimes P^A \in \mathbb{C}^{d^2 \times d^2}\,,
\end{equation}
which acts on the Choi state $\rho_{\mathcal{E}}$.
Let $\{\mathcal{W}_1, \dots, \mathcal{W}_{m}\}$ be independent Pauli observables, each of the form $\mathcal{W}_i = P_i^S \otimes P_i^A$, sampled uniformly at random from the bases set $\{\mathcal{W}_1, \dots, \mathcal{W}_{d^4}\}$.

Define the linear map $\Lambda: \mathbb{C}^{d^2 \times d^2} \to \mathbb{R}^{m}$ as
\begin{equation}\label{eq:Choimap}
(\Lambda(\rho_{\mathcal{E}}))_i = \sqrt{\frac{d^2}{m}} \Tr(\mathcal{W}_i \rho_{\mathcal{E}})\,,
\end{equation}
and consider the measurement model
\begin{equation}
\mathbf{y} = \Lambda(\rho_{\mathcal{E}}) + \mathbf{v} + \mathbf{z}\,,
\end{equation}
where $\mathbf{v}$ represetns the sparse corrupted noise with $\|\mathbf{v}\|_0 \le s$, and $\mathbf{z}$ is the measurement noise. We denote the extended matrix by ${\mathcal{M} = [\Lambda, I] \in \bC^{m \times (d^4 + m)}}$, and then derive the conditions under which $\mathcal{M}$ satisfies the generalized restricted isometry property.
\begin{proposition}\label{proposition1}
Let $\delta \in [0,1)$. If
\begin{equation}
m \ge C_1 \, r d^2 \log^6 d\,, \quad \text{and} \quad \sup \|\mathbf{v}\|_{\infty} \le \frac{\delta}{4d\sqrt{s}}\,,
\end{equation}
for some constant ${C_1 = O(1/\delta^2)}$, then with high probability, the extended measurement matrix ${\mathcal{M} = [\Lambda, I]}$ satisfies the GRIP with constant ${\delta_{r,s} \le \delta}$.
\end{proposition}
\begin{proof}
Similar to the proof in Appendix~\ref{app:RGstate}, the GRIP constant $\delta_{r,s}$ can be equivalently expressed as
\begin{equation}
  \delta_{r,s}=\sup_{(\rho_{\mathcal{E}},\bv)}\left|\biggl\|\cM\begin{bmatrix}
\rho_{\mathcal{E}} \\
\bv
\end{bmatrix}\biggr\|_2^2-\|\rho_{\mathcal{E}}\|_F^2-\|\bv\|_2^2\right|,
\end{equation}
where the supremum is taken over rank-$r$ matrices $\rho_{\mathcal{E}} \in \mathbb{C}^{d^2 \times d^2}$ and $s$-sparse vectors ${\mathbf{v} \in \mathbb{R}^{m}}$.
We split the expression as
\begin{equation}
  \delta_{r,s} \le \underbrace{\sup_{(\rho_{\mathcal{E}}, \mathbf{v})} \left| \|\Lambda(\rho_{\mathcal{E}})\|_2^2 - \|\rho_{\mathcal{E}}\|_F^2 \right|}_{\delta_1} + \underbrace{2 \sup_{(\rho_{\mathcal{E}}, \mathbf{v})} |\langle \Lambda(\rho_{\mathcal{E}}), \mathbf{v} \rangle|}_{\delta_2}.
\end{equation}

From the known RIP results, ${\delta_1 \le \delta/2}$ holds with high probability if ${m \ge C_1 r d^2 \log^6 d}$. For $\delta_2$, we note that ${|\Tr(\mathcal{W}_i \rho_{\mathcal{E}})| \le 1}$, and
\begin{equation}
  \sup\|\mathbf{v}\|_2 \le \sup\ \sqrt{s} \|\mathbf{v}\|_{\infty}\,,
\end{equation}
where ${\|\bv\|_{\infty}=\max_i|v_i|}$.
Therefore,
\begin{equation}
  \delta_2 \le 2 \sqrt{\frac{d^2}{m}} \cdot \sqrt{m s} \cdot \sup \|\mathbf{v}\|_{\infty} = 2d \sqrt{s} \sup \|\mathbf{v}\|_{\infty}\,.
\end{equation}
To ensure ${\delta_2 \le \delta/2}$, it suffices that
\begin{equation}
  \sup \|\mathbf{v}\|_{\infty} \le \frac{\delta}{4d \sqrt{s}}\,.
\end{equation}
Combining both bounds yields ${\delta_{r,s} \le \delta}$.
\end{proof}
\subsection{The estimator}
We can solve the following problem to recover the Choi state $\rho_{\mathcal{E}}$ and sparse noise $\bv$ \emph{simultaneously}:
\begin{equation}\label{eq:solveChoi}
  \begin{aligned} \minover[\rho^{\prime}_{\mathcal{E}},\bv^{\prime}] \quad &\frac{1}{2}\|\by-\tilde{\Lambda}(\rho^{\prime}_{\mathcal{E}})-\bv^{\prime}\|_{2}^2+\tau_1\cdot\|\rho^{\prime}_{\mathcal{E}}\|_{\tr}+\tau_2\cdot\|\bv^{\prime}\|_{1}\,,\\
   \subto \quad
      & \rho^{\prime}_{\mathcal{E}}\succeq 0\,, \Tr_S(\rho^{\prime}_{\mathcal{E}})=\frac{I}{d}\,,\\
      & \tau_1, \tau_2>0\,.
  \end{aligned}
\end{equation}
Here, $\tilde{\Lambda}(\rho^{\prime}_{\mathcal{E}})$ omits the factor $\sqrt{d^2/m}$ from Eq.~\eqref{eq:Choimap}, as this factor is not necessary for the estimator to work~\cite{PhysRevLett.132.240804}. The trace norm (or nuclear norm) of $X$ is defined as ${\|X\|_{\tr} = \Tr(\sqrt{X^{\dagger}X})}$, and $\rho^{\prime}_{\mathcal{E}}$ and $\mathbf{v}^{\prime}$ are the optimization variables. 
It is worth noting that the positive semidefinite constraint on the Choi state and the trace-preserving constraint of the quantum channel together lead to the Choi state having a constant trace norm. Therefore, the trace norm term in Eq.~\eqref{eq:solveChoi} can be omitted without significantly affecting the results~\cite{Kliesch2019}. The schematic framework is illustrated in Fig.~\ref{flow:Choi}.

To evaluate the reconstruction quality of the Choi state $\hat{\rho}_{\mathcal{E}}$ and the sparse noise $\hat{\bv}$, we adopt two standard measures. 
The first is the (squared) fidelity~\cite{PhysRevLett.106.230501,Jozsa1994}, defined as
\begin{equation}\label{fidelity}
F(\hat{\rho}_{\mathcal{E}},\rho_{\mathcal{E}}) = \left( \Tr \sqrt{ \hat{\rho}_{\mathcal{E}}^{1/2} \rho_{\mathcal{E}} \hat{\rho}_{\mathcal{E}}^{1/2} } \right)^2,
\end{equation}
which evaluates how close the estimated Choi state $\hat{\rho}_{\mathcal{E}}$ is to the true state $\rho_{\mathcal{E}}$.
The second measure is the mean squared error (MSE),
\begin{equation}\label{MSE}
T_{\text{MSE}} = \frac{1}{m} \sum_{i=1}^m (\bv_i - \hat{\bv}_i)^2,
\end{equation}
which quantifies the discrepancy between the estimated sparse noise $\hat{\bv}$ and the true noise $\bv$.

\subsection{Applications}\label{SecChoiapp}
Recall that the maximally entangled state is written as
${\ket{\Psi} = \frac{1}{\sqrt{d}} \sum_{j=0}^{d-1} \ket{j} \otimes \ket{j}}$, 
and the corresponding density operator is
${\varrho = \ket{\Psi}\bra{\Psi}}$. 
Let ${U \in \mathbb{C}^{d \times d}}$ be a unitary operator acting on the $n$-qubit system.
Then, the Choi state associated with the unitary channel $\mathcal{E}(\cdot) = U (\cdot) U^\dagger$ can be written as
\begin{equation}
  \rho_{\mathcal{E}} =
( U \otimes I_d )
\, \varrho \,
( U^\dagger \otimes I_d )\,,
\end{equation}
where $I_d$ denotes the identity operator on the $d$-dimensional Hilbert space.
Since the conclusion of Proposition~\ref{proposition1} is essentially a worst-case scenario, $\|\bv\|_\infty$ is bounded. In the numerical simulations, however, we adopt a more realistic model by using $s$-sparse Gaussian noise with zero mean and unit variance, where ${s = \lfloor \eta m \rfloor}$.

\begin{figure}[t]
\includegraphics[width=.95\columnwidth]{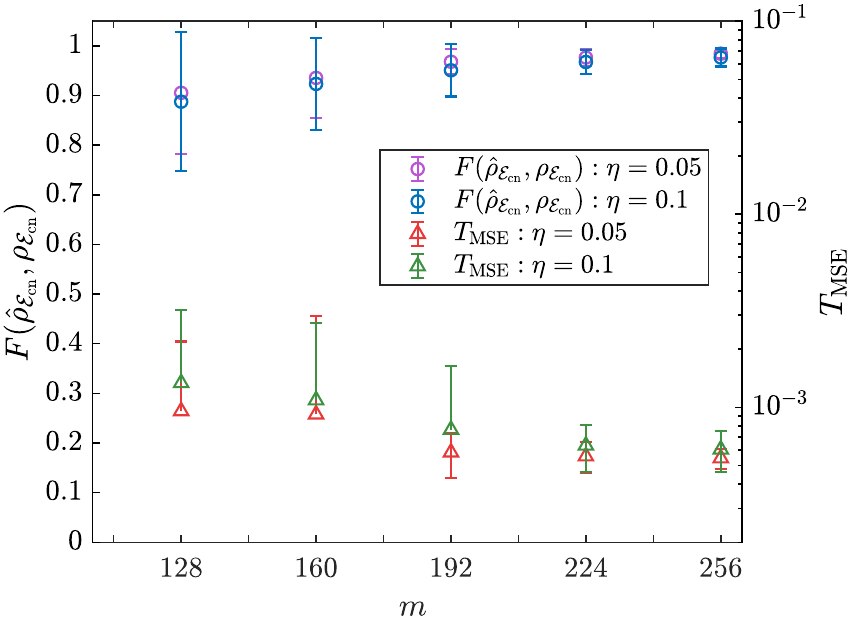}
    \caption{Fidelity $F(\hat{\rho}_{\mathcal{E}_{\rm cn}}, \rho_{\mathcal{E}_{\rm cn}})$ and MSE $T_{\text{MSE}}$ as functions of the number of sampled  Pauli observables $\mathcal{W}$ (denoted by $m$). The error bars are obtained from $100$ runs of tomography, with each run selecting $m$  Pauli observables randomly. The purple and blue  points represent the fidelity between the reconstructed Choi state and true Choi state of the CNOT gate, for sparsity ratios ${\eta = 0.05}$ and $0.1$, respectively. The red and green points show the MSE between the reconstructed sparse Gaussian noise and true sparse Gaussian noise, for sparsity ratios ${\eta = 0.05}$ and $0.1$, respectively. The Gaussian noise has a mean of zero, a standard deviation of $1$. The
regularization parameters are chosen as ${\tau_1 = 0.01m, \tau_2 = 10^{-2}}$.}
    \label{Fig_ChoiCNOT_ave_100_N_1000_eta_0.05_0.1_tex}
\end{figure}

For two-qubit gates, as illustrated in Fig.~\ref{Fig_ChoiCNOT_ave_100_N_1000_eta_0.05_0.1_tex} 
for the CNOT gate with sparsity ratios ${\eta = 0.05}$ and ${\eta = 0.1}$, the fidelity $F$ (and the mean squared error $T_{\text{MSE}}$) increase (decrease) with the number of sampled  Pauli observables $m$, thereby confirming the effectiveness of the  method. 
When ${\eta = 0.05}$, the overall performance in terms of fidelity and $T_{\mathrm{MSE}}$ is better than that in the case of ${\eta = 0.1}$. The measurement data are obtained using $10^3$ independent trials for each measurement setting, with finite-sampling noise absorbed into $\bz$.  However, the Choi-state scheme based on Pauli observables does not make effective use of the sampled observables.
Therefore, in Sec.~\ref{Sec:process} we conduct a more detailed numerical study using the process-matrix scheme. For instance, for the CNOT  gate, ${m = 192}$ is required to reach a fidelity of ${F(\hat{\rho}_{\mathcal{E}_{\rm cn}}, \rho_{\mathcal{E}_{\rm cn}})\approx 0.9506}$ under $\eta=0.1$, with an MSE of ${T_{\text{MSE}} \approx 7.66 \times 10^{-4}}$. In contrast, as shown in Sec.~\ref{2qgate}, a fidelity of ${F(\hat{\chi}_{\text{cn}}, \chi_{\text{cn}}) \approx 0.958}$ can be achieved using only $64$ configurations.

Beyond two-qubit gates, we turn to the three-qubit Toffoli gate. The Toffoli gate, also known as the CCNOT gate, inverts the target qubit (third qubit) if the first and second qubits are both $\ket{1}$, namely,
\begin{equation}
  U_{\text{Tof}}=\ket{0}\bra{0}\otimes I\otimes I+\ket{1}\bra{1}\otimes U_{\text{CNOT}}\,.
\end{equation}
For reference, Fig.~\ref{Fig_ChoiTof_ave_50_N_1000_eta_0.05_0.1} presents the results for the three-qubit Toffoli gate. When ${m \gtrsim 3072}$, the fidelity ${F(\hat{\rho}_{\mathcal{E}_{\rm Tof}}, \rho_{\mathcal{E}_{\rm Tof}})\gtrsim 0.95}$; however, even at ${m = 4096}$, the fidelity attains $0.9631$ under ${\eta = 0.05}$. Correspondingly, the process-matrix scheme is presented in Sec.~\ref{3qgate}.

\begin{figure}[t]
\includegraphics[width=.95\columnwidth]{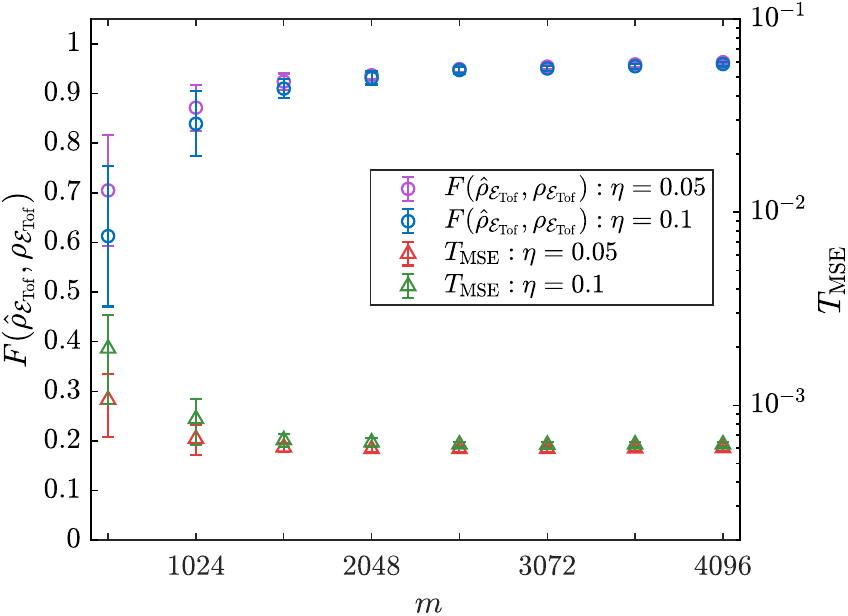}
    \caption{Fidelity $F(\hat{\rho}_{\mathcal{E}_{\rm Tof}}, \rho_{\mathcal{E}_{\rm Tof}})$ and MSE $T_{\text{MSE}}$ as functions of the number of sampled Pauli observables $\mathcal{W}$ (denoted by $m$). The error bars are obtained from $50$ runs of tomography, with each run selecting $m$ Pauli observables randomly. The purple and blue points represent the fidelity between the reconstructed Choi state and true Choi state of Toffoli gate, for sparsity ratios ${\eta = 0.05}$ and $0.1$, respectively. The red and green points show the MSE between the reconstructed sparse Gaussian noise and true sparse Gaussian noise, for sparsity ratios ${\eta = 0.05}$ and $0.1$, respectively. The Gaussian noise has a mean of zero, a standard deviation of $1$. The
regularization parameters are chosen as ${\tau_1 = 0.01m, \tau_2 = 10^{-2}}$.}
    \label{Fig_ChoiTof_ave_50_N_1000_eta_0.05_0.1}
\end{figure}
%

\section{Corrupted sensing quantum process tomography via process matrix}\label{Sec:process}

\subsection{The process-matrix representation} 
In the Kraus representation, a completely positive (CP) map $\cE$ is written as
\begin{equation}\label{KrauRepMap} \cE(\rho)=\sum_{i}E_i\rho E_i^{\dagger}\,,
\end{equation}
where $\{E_i\}$ are the Kraus operators. The map is trace preserving (TP) when ${\sum_i E_i^{\dagger}E_i=I}$. Usually it is more convenient to consider an equivalent description of $\cE$, where the process-matrix representation of $\cE$ can be obtained by writing the Kraus operators in a certain set of bases ${\Gamma_\alpha\in \mathbb{C}^{d\times d}}$ such that ${E_i=\sum_{\alpha=1}^{d^2} e_{i\alpha}\Gamma_\alpha}$, for some complex numbers $e_{i\alpha}$. Then the process-matrix representation can be expressed as
\begin{equation}\label{ProMat}
  \cE(\rho)=\sum_{\alpha,\beta=1}^{d^2}\chi_{\alpha\beta}\Gamma_{\alpha}\rho\Gamma_{\beta}^{\dagger}\,,
\end{equation}
where ${\chi_{\alpha\beta}=\sum_i e_{i\alpha}e^*_{i\beta}}$  are the elements of the ${d^2\times d^2}$ process matrix in the $\Gamma_\alpha$ bases for a system of $n$ qubits and ${d=2^n}$. The process matrix
is a positive-semidefinite matrix (implying that it is Hermitian), and satisfies the TP condition:
\begin{align}
    \chi\succeq 0\,,\\
\sum_{\alpha,\beta=1}^{d^2}\chi_{\alpha\beta}\Gamma_\beta^{\dagger}\Gamma_\alpha=I\,.
\end{align}
Hence, if the $\Gamma_\alpha$ bases fulfill ${\Tr(\Gamma_\beta^{\dagger}\Gamma_\alpha)=\delta_{\alpha\beta}}$, ${\Tr(\chi)=d}$, and ${\Tr(\Gamma_\beta^{\dagger}\Gamma_\alpha)=d\delta_{\alpha\beta}}$,  then ${\Tr(\chi)=1}$.

Given that any representation of a general map has ${d^2(d^2-1)}$ independent parameters, the specification of such a map can be accomplished by probing it with $d^2$ input states $\rho^{\text{in}}$, forming a set of Hermitian operator bases. An informationally complete measurement on the resulting output state $\cE(\rho^{\text{in}})$ reveals ${d^2-1}$ independent parameters that characterize the map.

\begin{figure*}[t]
\includegraphics[width=0.7\textwidth]{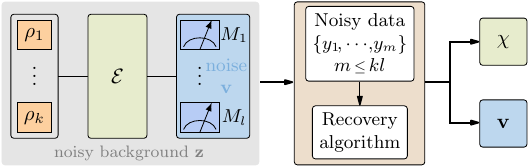}
    \caption{Schematic procedure of the corrupted sensing QPT based on the process-matrix representation. Consider preparing different initial states $\rho_o$, applying the process $\cE$, and then measuring the outputs using a set of observables $M_o$ to get the noisy data $\by\!:\{y_1,\cdots,y_m\}$. Here, we select $m$ configurations $(\rho_o,M_o), o= 1,\cdots, m\leq kl$ randomly. By using the recovery algorithm, the process matrix $\chi$ and the measurement noise $\bv$ can be reconstructed simultaneously, and $\bz$ denotes any other potential noise.}
    \label{flow}
\end{figure*}
\subsection{Problem formulation}
Consider preparing different initial input states $\{\rho_1,\cdots,\rho_k\}$, applying the process  $\cE$, and then measuring the outputs using the observable $M$ from the set ${\{M_1,\cdots,M_l\}}$. For a pair  ${(\rho_o,M_o), o=1,\cdots,m\leq kl}$, the presence of noise in the measurement process might corrupt the measurement data. Consequently, the outcome can be expressed as
\begin{equation}\label{eq:model}
\by=\Phi\bc+\bv+\bz\,,
\end{equation}
where ${\bm{\chi}\in \mathbb{C}^{d^4\times 1}}$ is the vectorized form of the process matrix $\chi$ in the $\{\Gamma_\alpha\}$ bases, $\Phi$ is an ${m\times d^4}$ matrix with elements of the form ${\Tr(\Gamma_{\alpha}\rho_o\Gamma_{\beta}^{\dagger}M_o)/\sqrt{m}}$.
And the corrupted noise is considered as a structured vector ${\bv\in \mathbb{R}^m}$, and ${\bz\in \mathbb{R}^m}$
is any other kind of unstructured noise.

If we assume the process matrix $\chi$ being sparse in any known basis, and the noise as a sparse vector $\bv$, the following problem can be solved to recover  $\chi$ and $\bv$ \emph{simultaneously}:
\begin{equation}\label{eq:sparseChiV}
  \begin{aligned}
    \minover[\chi^{\prime},\bv^{\prime}] \quad &\frac{1}{2}\|\by-\Phi\bc^{\prime}-\bv^{\prime}\|_{2}^2+\mu_1\cdot\|\bc^{\prime}\|_{1}+\mu_2\cdot\|\bv^{\prime}\|_{1}\,,\\
   \subto \quad
      & \chi^{\prime}\succeq 0\,, \sum_{\alpha,\beta}\chi^{\prime}_{\alpha\beta}\Gamma_\beta^{\dagger}\Gamma_\alpha=I\,,\\
      & \mu_1, \mu_2>0\,.
  \end{aligned}
\end{equation}
The schematic framework is shown in Fig.~\ref{flow}.

\subsection{Two-qubit entangling gates}\label{2qgate}
Consider encoding the qubits in polarization states with ${\ket{H}=\ket{0}}$, ${\ket{V}=\ket{1}}$, and $\ket{D(A)}=(\ket{H}\pm\ket{V})/\sqrt{2}$, $\ket{R(L)}=(\ket{H}\pm i\ket{V})/\sqrt{2}$. 
For a two-qubit gate, standard process tomography requires $256$ different settings of input
states and measurement projectors. There is significant flexibility in selecting the tomographic input and measurement settings. As one accessible choice, our input states $\rho$ are chosen from $\{HH, VH, DV, RH, RV, VV, HV, HA,  HR, RR, RA, DA, \allowbreak DL, VA, VR, DH\}$ and the measurements $M$ are from $\{HH, HV, VH, VV, HD, HR, VR, VD, DD, LR, LD, DL, \allowbreak DV, LV, DH, LH\}$. In total, there are ${16\times 16=256}$ configurations of the $(\rho,M)$ pairs.

We employ two approaches for the numerical calculations using MATLAB: the results for the two-qubit gates in this section are obtained using the \textsc{cvx}~\cite{cvx} package with the solver \textsc{sdpt3}, while for larger scales, the results presented in Secs.~\ref{3qgate} and~\ref{4qgate} use the \textsc{cvx} package with the solver \textsc{mosek} to achieve more stable solutions. Here, by randomly selecting ${m\leq 256}$ configurations from all combinations, the process matrix $\chi$ and corrupted noise $\bv$ can be obtained simultaneously by solving Eq.~\eqref{eq:sparseChiV}. The process matrices in the Pauli bases for the CNOT, CZ, and SWAP gates are given in Appendix~\ref{app:2q_matrix}. Using the Pauli bases set $\{\Gamma_{\alpha}^{\text{Pau}}\}$, we assess the reconstruction of each gate by means of the fidelity $F(\hat{\chi},\chi)$ and the mean squared error $T_{\text{MSE}}$, as shown in Fig.~\ref{Fig_CNOT_Gau_ave_100_N_1000} and Fig.~\ref{Fig_CZ_SWAP_supplement_tex} in Appendix~\ref{app:2q_matrix}.

For the CNOT gate, a fidelity of ${F(\hat{\chi}_{\text{cn}}, \chi_{\text{cn}}) \approx 0.962}$ is achieved with ${m = 64}$ configurations, corresponding to an MSE of ${T_{\text{MSE}} \approx 2 \times 10^{-4}}$. Similarly, the two-qubit CZ gate requires only ${m = 64}$ configurations to achieve ${F(\hat{\chi}_{\text{cz}},\chi_{\text{cz}}) \approx 0.971}$, corresponding to $T_{\text{MSE}}\approx 1.368\times 10^{-4}$. Finally, the SWAP gate reaches ${F(\hat{\chi}_{\text{sw}}, \chi_{\text{sw}})\approx 0.978}$ at ${m=48}$ (i.e., a sampling rate of $18.75\%$), with $T_{\text{MSE}}\approx 1.55\times 10^{-4}$.

\begin{figure}[t]
\includegraphics[width=.95\columnwidth]{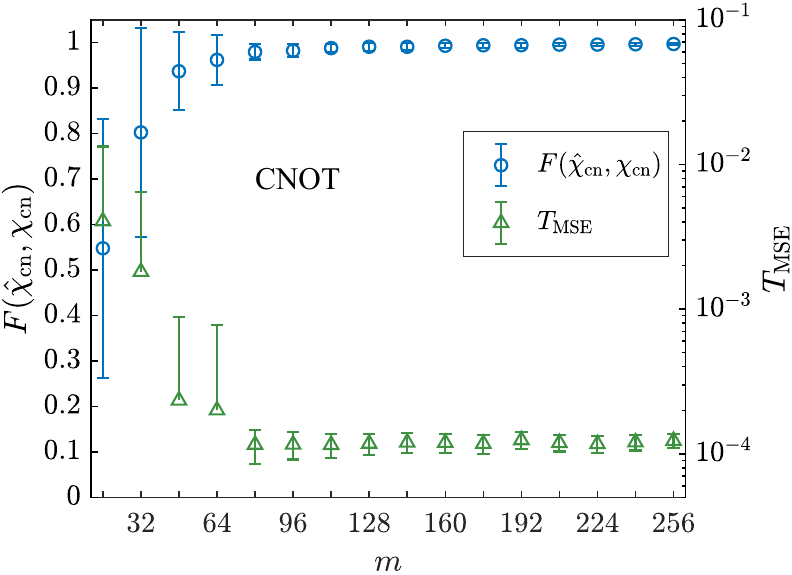}
    \caption{Fidelity $F(\hat{\chi}_{\text{cn}},\chi_{\text{cn}})$ and MSE $T_{\text{MSE}}$ as functions of the number of configurations $m$. The error bars are obtained from $100$ runs of tomography, with each run selecting $m$ combinations randomly. The blue points represent the fidelity between the reconstructed process matrix and true process matrix of the CNOT gate. The green points show the MSE between the reconstructed sparse Gaussian noise and true sparse Gaussian noise. The Gaussian noise has a mean of zero, a standard deviation of $1$, and a sparsity level of ${s = \lfloor 0.1m \rfloor}$. The
regularization parameters are chosen as ${\mu_1 = 10^{-5}, \mu_2 = 10^{-3}}$.}
    \label{Fig_CNOT_Gau_ave_100_N_1000}
\end{figure}
\subsection{Three-qubit gates}\label{3qgate}
In this section, we prepare pairwise combinations of the $64$ inputs $\{\ket{H},\ket{V},\ket{D},\ket{R}\}^{\otimes 3}$ and $64$ observables $\{\ket{H},\ket{V},\ket{D},\ket{R}\}^{\otimes 3}$, and then randomly select ${m\leq 4096}$ configurations from all ${d^4=4096}$ possible combinations. The process matrix in the Pauli bases of the Toffoli gate is given in Appendix~\ref{app:3q_matrix}. With the Pauli bases set $\{\Gamma_{\alpha}^{\text{Pau}}\}$, we assess the reconstruction by means of $F(\hat{\chi}_{\text{Tof}},\chi_{\text{Tof}})$ and $T_{\text{MSE}}$; see Fig.~\ref{Fig_Tof_Gau_ave_50_N_1000}.

\begin{figure}[t]
    \includegraphics[width=.95\columnwidth]{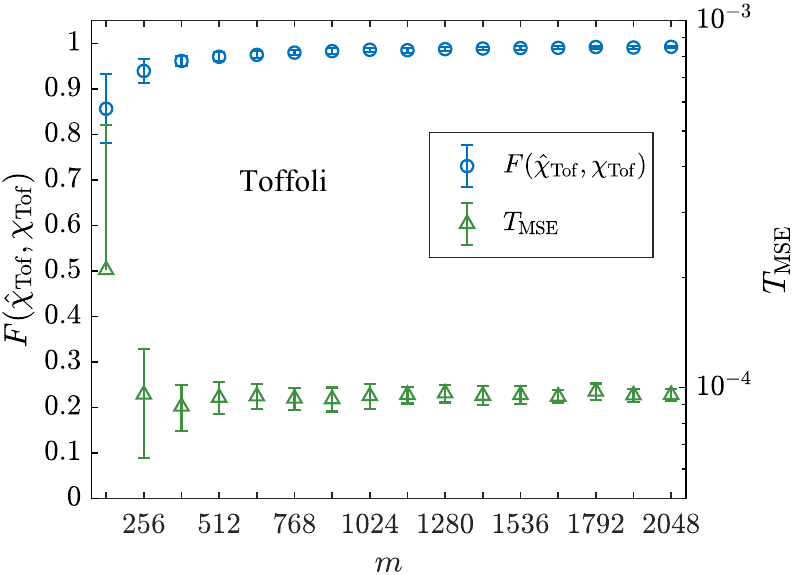}
    \caption{Fidelity $F(\hat{\chi}_{\text{Tof}},\chi_{\text{Tof}})$ and MSE $T_{\text{MSE}}$ as functions of the number of configurations $m$. The error bars are obtained from $50$ runs of tomography, with each run selecting $m$ combinations randomly. The blue points represent the fidelity between the reconstructed process matrix and true process matrix of the Toffoli gate. The green points show the MSE between the reconstructed sparse Gaussian noise and true sparse Gaussian noise. The Gaussian noise has a mean of zero, a standard deviation of $1$, and a sparsity level of ${s = \lfloor 0.1m \rfloor}$. The
regularization parameters are chosen as ${\mu_1 = 10^{-5}, \mu_2 = 10^{-3}}$.}
    \label{Fig_Tof_Gau_ave_50_N_1000}
\end{figure}

The three-qubit Toffoli gate requires lower sampling compared to two-qubit gates to achieve comparable fidelities: a fidelity of ${F(\hat{\chi}_{\text{Tof}},\chi_{\text{Tof}}) \approx 0.9615}$ is already achieved at ${m=384}$ (i.e., a sampling rate of $ 9.375\%$), with an MSE of $T_{\text{MSE}}\approx 8.89\times 10^{-5}$, highlighting the scalability of the method under increased system size. Furthermore, when $m = 1664$ (corresponding to a sampling rate of $40.625\%$), the fidelity reaches $0.99$.

The Fredkin gate or the three-qubit conditional SWAP gate swaps the last two qubits if the first qubit is $\ket{1}$, namely,
\begin{equation}
  U_{\text{Fred}}=\ket{0}\bra{0}\otimes I\otimes I+\ket{1}\bra{1}\otimes U_{\text{SWAP}}\,.
\end{equation}
The process matrix in the Pauli bases of the Fredkin gate is given in Appendix~\ref{app:3q_matrix}. With the Pauli bases set $\{\Gamma_{\alpha}^{\text{Pau}}\}$, we assess the reconstruction by means of $F(\hat{\chi}_{\text{Fred}},\chi_{\text{Fred}})$ and $T_{\text{MSE}}$; see Fig.~\ref{Fig_Fred_Gau_ave_50_N_1000} in Appendix~\ref{app:3q_matrix}.

The process matrix of the Fredkin gate is reconstructed with a fidelity of ${F(\hat{\chi}_{\text{Fred}},\chi_{\text{Fred}})\approx 0.9553}$ at ${m=256}$ (i.e., a sampling rate of $ 6.25\%$), accompanied by an MSE of ${T_{\text{MSE}}\approx 9.2\times 10^{-5}}$. The fidelity is further improved to ${F(\hat{\chi}_{\text{Fred}},\chi_{\text{Fred}}) \approx 0.99}$ when the sampling rate reaches $37.5\%$.

\subsection{Four-qubit CCCZ gate}\label{4qgate}
The $x$-controlled $Z$ gate, also denoted as $C^xZ$, applies a $Z$-gate to the target qubit controlled by $x$ qubits. For ${x=3}$, it is commonly referred to as the CCCZ gate~\cite{PhysRevA.91.032311, song2017continuous, huang2024demonstration}, a typical four-qubit gate, given by the expression:
\begin{equation}
  U_{\text{CCCZ}}=I-2\ket{1111}\bra{1111}\,.
\end{equation}

In this section, we prepare pairwise combinations of the ${d^2=256}$ inputs $\{\ket{H},\ket{V},\ket{D},\ket{R}\}^{\otimes 4}$ and $d^2$ observables $\{\ket{H},\ket{V},\ket{D},\ket{R}\}^{\otimes 4}$, and then randomly select ${m\ll d^4}$ configurations from all possible combinations.
The process matrix of the CCCZ gate in the Pauli bases is given in Appendix~\ref{app:CCCZ}. With the Pauli bases set $\{\Gamma_{\alpha}^{\text{Pau}}\}$, we assess the reconstruction by means of $F(\hat{\chi}_{\text{cccz}},\chi_{\text{cccz}})$ and $T_{\text{MSE}}$; see Fig.~\ref{Fig_CCCZ_Gau}.

The four-qubit CCCZ gate requires the lowest sampling: a sampling rate of ${3.125\%}$ (${m=2048}$) yields a fidelity of ${F(\hat{\chi}_{\text{cccz}},\chi_{\text{cccz}}) \approx 0.971}$, with an MSE of  ${T_{\text{MSE}}\approx 5.77\times 10^{-5}}$. The fidelity can be improved to ${F(\hat{\chi}_{\text{cccz}},\chi_{\text{cccz}}) \approx 0.9945}$ at ${m=16384}$ (i.e., a sampling rate of $25\%$),  which further demonstrates the method’s scalability to larger systems.

\begin{figure}[t]
\includegraphics[width=.95\columnwidth]{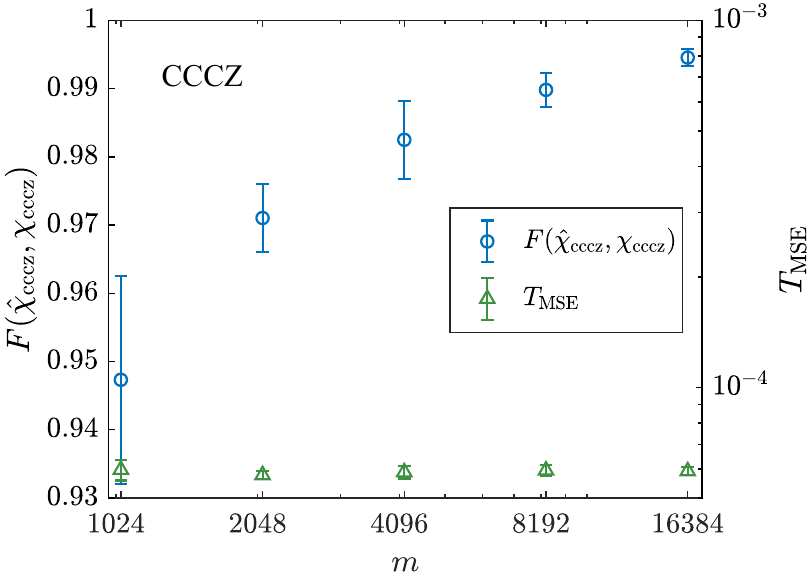}
    \caption{Fidelity $F(\hat{\chi}_{\text{cccz}},\chi_{\text{cccz}})$ and MSE $T_{\text{MSE}}$ as functions of the number of configurations $m$. The error bars are obtained from $10$ runs of tomography, with each run selecting $m$ combinations randomly. The blue points represent the fidelity between the reconstructed process matrix and true process matrix of the CCCZ gate. The green points show the MSE between the reconstructed sparse Gaussian noise and true sparse Gaussian noise. The Gaussian noise has a mean of zero, a standard deviation of $1$, and a sparsity level of ${s = \lfloor 0.1m \rfloor}$. The
regularization parameters are chosen as ${\mu_1 = 10^{-5}, \mu_2 = 10^{-3}}$.}
    \label{Fig_CCCZ_Gau}
\end{figure}
%

\section{Summary}\label{Sec:conclusion}
In this work, we have proposed a corrupted sensing framework for quantum process tomography that enables the simultaneous reconstruction of quantum processes and sparse measurement noise. By introducing the generalized restricted isometry property for the Choi-state representation, we established a probabilistic condition under which the extended measurement matrix satisfies the GRIP with high probability. 

Furthermore, systematic numerical simulations of several typical quantum gates were conducted in the process-matrix representation to validate the effectiveness of the method. In comparison with the Choi-state scheme, which selects random Pauli operators, the process-matrix approach utilizes the selected input–output measurement configurations more efficiently. The results show that the four-qubit CCCZ gate requires the lowest sampling rate to reach the fidelity threshold, highlighting the scalability of the method in handling larger multi-qubit systems.

Several aspects merit further investigation. First, based on the established GRIP condition, deriving explicit error bounds would offer a more complete theoretical guarantee for the proposed corrupted sensing framework. In addition, since the GRIP is a strong property for low-rank matrices and sparse noise, alternative analytical tools may be explored to provide theoretical insights. Finally, measurement tomography in the corrupted sensing setting can be considered.

\acknowledgments
We thank Yinfei Li and Ye-Chao Liu for helpful discussions. We are also grateful to Yulong Liu, Yanwu Gu and Zhenpeng Xu for inspiring discussions at an earlier stage of this project.
This work was supported by the National Natural Science Foundation of China (Grants No.~92265115 and No.~12175014) and the National Key R\&D Program of China (Grant No.~2022YFA1404900).


%

%
\appendix
\onecolumngrid

\section{The GRIP condition for corrupted sensing quantum state tomography}\label{app:RGstate}
Consider an $n$-qubit quantum system with dimension ${d=2^n}$, the unknown state of the system is denoted by $\rho$, which satisfies ${\Tr(\rho)=1}$ and ${\rho\succeq 0}$.
An $n$-qubit Pauli operator takes on the general form
\begin{equation}
  W=\bigotimes_{i=1}^n\sigma_i\,,
\end{equation}
where ${\sigma_i\in\{I, \sigma_x, \sigma_y, \sigma_z\}}$. Here, $\sigma_x, \sigma_y, \sigma_z$ are the three Pauli matrices, and $I$ represents the identity matrix. The density matrix $\rho$ can be expanded by using the Pauli observables ${\{W_i\!:i\in[d^2]\}}$.
By choosing $m$ elements $\{P_1, P_2, \cdots , P_{m}\}$ uniformly at random from the Pauli bases set $\{W_1,W_2,\cdots,W_{d^2}\}$, the $i$-th component $(\cA(\rho))_i$ of the linear map ${\cA\!:\bC^{d\times d}\to \bR^{m}}$ for Pauli measurements can be defined as
\begin{equation}\label{Amap}
  (\cA(\rho))_i=\sqrt{\frac{d}{m}}\Tr(P_i\rho)
\end{equation}
for ${\rho\in \bC^{d\times d}}$. The corresponding adjoint operator ${\cA^{\dagger}\!:\bR^{m}\to \bC^{d\times d}}$ is
\begin{equation}\label{2}
  \cA^{\dagger}(\by)=\sqrt{\frac{d}{m}}\sum_{i=1}^{m}y_iP_i\,,
\end{equation}
where ${\by\in \bR^{m}}$. The coefficient ${\sqrt{d/{m}}}$ is chosen so that ${\bE[\cA^{\dagger}(\cA(\rho))]=\rho}$.
\begin{proposition}
  Suppose $\by=\cA(\rho)+\bv+\bz$ with the extended matrix $\cM=[\cA,I]\in \bC^{{m}\times (d+{m})}$. For ${\delta\in[0,1)}$, if
  \begin{equation}\label{mbound}
    m\ge C_2rd\log^6d\,, \quad \text{and} \quad \sup \|\bv\|_{\infty}\leq \frac{\delta}{4\sqrt{ds}}\,,
  \end{equation}
for some constant ${C_2=O(1/\delta^2)}$, then with high probability, $\mathcal{M}$ has the GRIP constant satisfying $\delta_{r, s}\leq \delta$.
\end{proposition}
\begin{proof}
The GRIP constant $\delta_{r, s}$ can be equivalently expressed as
\begin{equation}
  \delta_{r,s}=\sup_{(\rho,\bv)\in\mathcal{T}}\left|\biggl\|\cM\begin{bmatrix}
\rho \\
\bv
\end{bmatrix}\biggr\|_2^2-\|\rho\|_F^2-\|\bv\|_2^2\right|,
\end{equation}
where $\mathcal{T}\!:=\{(\rho,\bv)\!: \text{rank}(\rho)\leq r,\|\bv\|_0\le s, \|\rho\|_F^2+\|\bv\|_2^2=1,   \rho\in\bC^{d\times d},\bv\in \bR^{m}\}$. With ${\cM=[\cA,I]}$, $\biggl\|\cM\begin{bmatrix}
\rho \\
\bv
\end{bmatrix}\biggr\|_2^2=\|\cA (\rho)\|_2^2+\|\bv\|_2^2+2\langle\mathcal{A}(\rho), \bv\rangle$, the GRIP constant can be written as
\begin{equation}
 \begin{aligned}
  \delta_{r,s}&=\sup _{(\rho, \bv) \in \mathcal{T}}\left|\|\cA (\rho)\|_2^2+\|\bv\|_2^2+2\langle\mathcal{A}(\rho), \bv\rangle-\|\rho\|_F^2-\|\bv\|_2^2\right|\\
  &\le \underbrace{\sup _{(\rho, \bv) \in \mathcal{T}}\left|\|\cA (\rho)\|_2^2-\|\rho\|_F^2\right|}_{\delta_1}+\underbrace{2 \sup _{(\rho, \bv) \in \mathcal{T}}|\langle\mathcal{A} (\rho), \bv\rangle|}_{\delta_2}.
  \end{aligned}
\end{equation}
The aim is to derive a bound for the number of measurements $m$ such that for any $\delta\in[0,1)$, the GRIP constant $\delta_{r,s}$ is upper bounded by $\delta$.

Recall the RIP definition as written in Eq.~\eqref{supRIP}, we have
\begin{equation}
  \delta_1=\sup_{\rho\in \cD_r}\left|\|\cA (\rho)\|_2^2-\|\rho\|_F^2 \right|\leq \delta\,.
\end{equation}
Considering Theorem 2.1 of Ref.~\cite{Liu}, we
fix the constant ${\delta\in[0,1)}$, and let ${m=C\cdot rd\log^6d}$ for some constant ${C=O(1/\delta^2)}$ depending solely on $\delta$. Then, over the choice of $\{P_1,\cdots,P_{m}\}$, the map $\cA$ satisfies the RIP with high probability over the set of all ${X\in\bC^{d\times d}}$ such that ${\|X\|_{\tr}\leq \sqrt{r}\|X\|_F}$. Here, the set of all ${X\in\bC^{d\times d}}$ with ${\|X\|_{\tr}\leq \sqrt{r}\|X\|_F}$ is slightly larger than the set of all $X\in\bC^{d\times d}$ with rank $r$. Furthermore, the failure probability is exponentially small in $\delta^2C$.
Then, for $\delta\in[0,1)$,
${m=4Crd\log^6d}$ for some constant ${C=O(1/\delta^2)}$, ${\delta_1\leq \delta/2}$ holds with high probability. Therefore, proof of the GRIP condition reduces to bounding the term $\delta_2$.

Now consider $\rho\succeq 0,\Tr(\rho)=1$, thus $\cA(\rho)$ is a real vector,
 \begin{equation}
 \begin{aligned}
\delta_2&=2 \sup _{(\rho, \bv) \in \cT_q}|\langle\cA (\rho), \bv\rangle|=2 \sqrt{\frac{d}{m}}\sup _{(\rho, \bv) \in \cT_q}\left|\sum_{i=1}^{m}\Tr(P_i\rho)v_i\right|\\
&=2\sqrt{\frac{d}{m}} \sup _{(\rho, \bv) \in \cT_q}\left|\langle\bu,\bv\rangle\right|,
\end{aligned}
\end{equation}
where $\bu=\bu(\rho)=(\Tr(P_1\rho),\Tr(P_2\rho),\cdots,\Tr(P_{m}\rho))^{\top}$, and ${u_i=\Tr(P_i\rho)\in[-1,1]}$; $\cT_q = \{(\rho, \bv) \!: \rho \in \cX, \bv \in \cV\}$, and $\cX = \{ \rho\in \bC^{d \times d} : \rho \succeq 0, \Tr(\rho) = 1, \text{rank}(\rho) \leq r \}
$, $\cV = \{ \bv\in \bR^{m} \!: \|\bv\|_0 \leq s\}
$.
Note that $\sup _{\rho \in\cX}\|\bu\|_2=\sup _{\rho \in\cX}\sqrt{\sum_{i=1}^{m}u_i^2}\leq\sqrt{\sum_{i=1}^{m}1}=\sqrt{m}$, and 
\begin{equation}\label{eq:v2norm}
\sup_{\bv\in\cV} \|\bv\|_2
= \sup_{\bv\in\cV} \sqrt{\sum_{i \in \operatorname{supp}(\bv)} v_i^2}
\le \sqrt{s}\; \sup_{\bv\in\cV} \|\bv\|_\infty\,,
\end{equation}
where $\operatorname{supp}(\bv) := \{\, j \in \{1,\dots,m\} : v_j \neq 0 \,\}$ denotes the set of indices corresponding to the nonzero entries of $\bv$, and ${\|\bv\|_{\infty}=\max_i|v_i|}$. From the Cauchy-Schwarz inequality, we can get
\begin{equation}\label{eq:CSIn}
 \begin{aligned}
  \sup _{(\rho, \bv) \in \cT_q}\left|\langle\bu,\bv\rangle\right|
  &\leq\sup _{\rho \in\cX}\|\bu\|_2\cdot\sup_{\bv\in\cV}\|\bv\|_2\\
  &\leq \sup\sqrt{ms}\|\bv\|_{\infty}\,.
 \end{aligned}
\end{equation}
Therefore,
\begin{equation}
 \begin{aligned}
\delta_2&=2\sqrt{\frac{d}{m}} \sup _{(\rho, \bv) \in \cT_q}\left|\langle\bu,\bv\rangle\right|\\
&\leq \sup\ 2\sqrt{ds} \|\bv\|_{\infty}\,.
 \end{aligned}
\end{equation}
To ensure $\delta_2\leq \delta/2$,  it suffices to impose
\begin{equation}
 \sup\ \|\bv\|_{\infty}\leq \frac{\delta}{4\sqrt{ds}}\,.
\end{equation}
Combining $\delta_1\leq \delta/2$ and $\delta_2\leq \delta/2$, we get $\delta_{r,s}\leq \delta$, concluding the proof.
\end{proof}
\begin{remark}
The GRIP condition can be used to prove the recovery guarantee. Note that the upper bound on the noise in Eq.~\eqref{mbound} may correspond to a worst-case scenario. The actual recovery error also depends on the choice of estimator and its associated reconstruction strategy. Furthermore, when analyzing the condition
$\delta_2\leq c\delta$, one can also obtain an alternative bound in a probabilistic sense by employing techniques such as covering numbers.
\end{remark}

\section{Process matrices of various quantum gates}
In this section, we present explicit process-matrix representations of various quantum gates under the Pauli bases. In addition, the reconstruction results for the CZ and SWAP gates are shown in Fig.~\ref{Fig_CZ_SWAP_supplement_tex}, while the reconstruction results for the Fredkin gate are presented in Fig.~\ref{Fig_Fred_Gau_ave_50_N_1000}. These results serve as references for the numerical simulations and reconstructions discussed in the main text.

\subsection{Two-qubit CNOT, CZ, and SWAP gates}\label{app:2q_matrix}
The matrix form of two-qubit CNOT gate is represented as follows:
\begin{equation}
U_{\text{CNOT}}=\left[
\begin{array}{cccc}
    1&0&0&0\\
    0&1&0&0\\
    0&0&0&1\\
    0&0&1&0
\end{array}
\right]\!.
\end{equation}
Recall that the process matrix of a channel $\cE(\rho)=\sum_{i}E_i\rho E_i^{\dagger}=\sum_{\alpha,\beta=1}^{d^2}\chi_{\alpha\beta}\Gamma_{\alpha}\rho\Gamma_{\beta}^{\dagger}$.
The CNOT gate is unitary, so the corresponding channel ${\mathcal{E}_{\text{CNOT}}(\rho)=U_{\text{CNOT}}\rho U_{\text{CNOT}}^\dagger}$ has a single Kraus operator. Therefore, $U_{\text{CNOT}}=\sum_{\alpha=1}^{d^2} e_{\alpha}\Gamma_\alpha, \chi_{\alpha\beta}=e_{\alpha}e^*_{\beta}$.

Using Pauli bases representation, the CNOT gate can be written as ${U_{\text{CNOT}}=\frac{1}{2}(I\otimes I+I\otimes X+Z\otimes I-Z\otimes X)}$ of the tensor products of Pauli operators $\{I,X,Y,Z\}$~\cite{PhysRevLett.93.080502}. Consequently, the $16\times 16$ process matrix has $16$ non-zero elements, see Table~\ref{tab:CNOTPB}.
\begin{table}[t]
\centering
\setlength{\extrarowheight}{4pt}
\scalebox{1.1}{
\begin{tabular}{ *{6}{c} }
\toprule
& $\bm{II}$ & $\bm{IX}$ & $\bm{ZI}$ & $\bm{ZX}$ & \textbf{others} \\
\midrule
$\bm{II}$ & $\invplus\frac{1}{4}$ & $\invplus\frac{1}{4}$ & $\invplus\frac{1}{4}$ & $-\frac{1}{4}$ & $0$ \\
$\bm{IX}$ & $\invplus\frac{1}{4}$ & $\invplus\frac{1}{4}$ & $\invplus\frac{1}{4}$ & $-\frac{1}{4}$ & $0$ \\
$\bm{ZI}$ & $\invplus\frac{1}{4}$ & $\invplus\frac{1}{4}$ & $\invplus\frac{1}{4}$ & $-\frac{1}{4}$ & $0$ \\
$\bm{ZX}$ & $-\frac{1}{4}$ & $-\frac{1}{4}$ & $-\frac{1}{4}$ & $\invplus\frac{1}{4}$ & $0$ \\
\textbf{others} & $\invplus0$ & $\invplus0$ & $\invplus0$ & $\invplus0$ & $0$ \\
\bottomrule
\end{tabular}
}
    \caption{Process matrix elements of the CNOT gate under Pauli bases. The order of bases is as follows: $II,\allowbreak IX,\allowbreak IY,\allowbreak IZ,\allowbreak XI,\allowbreak XX,\allowbreak XY,\allowbreak XZ,\allowbreak YI,\allowbreak YX,\allowbreak YY,\allowbreak YZ,\allowbreak ZI,\allowbreak ZX,\allowbreak ZY,\allowbreak ZZ$.}
    \label{tab:CNOTPB}
\end{table}
Note that ${\Tr(\Gamma_{\alpha}^{\dagger}\Gamma_{\beta})=d\delta_{\alpha\beta}}$ if ${\{\Gamma_{\alpha}\in \mathbb{C}^{d\times d}\}}$ are the Pauli bases, for a trace-preserving operation ${\Tr(\chi^{\text{Pau}})=1}$.

The form of the process matrix is dependent on the choice of bases. We also provide the representations in the computational and singular value decomposition (SVD) bases for reference.
Using the computational bases representation, the CNOT can be written as $U_{\text{CNOT}}=\ket{00}\bra{00}+\ket{01}\bra{01}+\ket{10}\bra{11}+\ket{11}\bra{10}$. Consequently, the process matrix has $16$ non-zero elements, see Table~\ref{tab:CNOTCB}.
\begin{table}[t]
\centering
\setlength{\extrarowheight}{4pt}
\scalebox{1.1}{
\begin{tabular}{ *{6}{c} }
\toprule
& $\bm{\ket{00}\bra{00}}$ & $\bm{\ket{01}\bra{01}}$ & $\bm{\ket{10}\bra{11}}$ & $\bm{\ket{11}\bra{10}}$ & \textbf{others} \\
\midrule
$\bm{\ket{00}\bra{00}}$ & $\invplus 1$ & $\invplus 1$ & $\invplus 1$ & $\invplus 1$ & $0$ \\
$\bm{\ket{01}\bra{01}}$ & $\invplus 1$ & $\invplus 1$ & $\invplus 1$ & $\invplus 1$ & $0$ \\
$\bm{\ket{10}\bra{11}}$ & $\invplus 1$ & $\invplus 1$ & $\invplus 1$ & $\invplus 1$ & $0$ \\
$\bm{\ket{11}\bra{10}}$ & $\invplus 1$ & $\invplus 1$ & $\invplus 1$ & $\invplus 1$ & $0$ \\
\textbf{others} & $\invplus 0$ & $\invplus 0$ & $\invplus 0$ & $\invplus 0$ & $0$ \\
\bottomrule
\end{tabular}
}
    \caption{Process matrix elements of the CNOT gate under the computational bases.}
    \label{tab:CNOTCB}
\end{table}
Note that ${\Tr(\Gamma_{\alpha}^{\dagger}\Gamma_{\beta})=\delta_{\alpha\beta}}$ if ${\{\Gamma_{\alpha}\in \mathbb{C}^{d\times d}\}}$ are the computational bases, for a trace-preserving operation ${\Tr(\chi^{\text{Com}})=d}$.

The representation under the SVD bases can be obtained through the following steps. First, perform SVD on the process matrix in the computational bases $\chi^{\text{Com}}$:
\begin{equation}
    \chi^{\text{Com}}=V\text{diag}(d,0,\cdots,0)V^{\dagger}\,,
\end{equation}
where $V$ is a $d^2\times d^2$ unitary matrix.
Then, the SVD-bases matrices can be obtained by
\begin{equation}
    \{\Gamma_{\alpha}^{\text{SVD}}=\sum_{\beta=1}^{d^2}V_{\beta\alpha}\Gamma_{\beta}^{\text{Com}}\in\cC^{d\times d}\}_{\alpha=1}^{d^2}\,,
\end{equation}
where $\Gamma_{\beta}^{\text{Com}}\in\cC^{d\times d}$ is the computational bases. The CNOT can be represented as $U_{\text{CNOT}}=2\Gamma_1^{\text{SVD}}$. Therefore, the process matrix has only one non-zero element ${\chi^{\text{SVD}}_{11}=4}$.
Note that ${\Tr(\Gamma_{\alpha}^{\dagger}\Gamma_{\beta})=\delta_{\alpha\beta}}$ if ${\{\Gamma_{\alpha}\in \mathbb{C}^{d\times d}\}}$ are the SVD bases, for a trace-preserving operation ${\Tr(\chi^{\text{SVD}})=d}$.

The matrix form of the CZ gate is represented as follows:
\begin{equation}
U_{\text{CZ}}=\left[
\begin{array}{cccc}
    1&0&0&0\\
    0&1&0&0\\
    0&0&1&0\\
    0&0&0&-1
\end{array}
\right]\!.
\end{equation}
Using Pauli bases representation, the CZ gate can be written as ${U_{\text{CZ}}=\frac{1}{2}(I\otimes I+I\otimes Z+Z\otimes I-Z\otimes Z)}$ of the tensor products of Pauli operators  $\{I,X,Y,Z\}$; see Table~\ref{tab:CZPB}.
\begin{table}[t]
\centering
\setlength{\extrarowheight}{4pt}
\scalebox{1.1}{
\begin{tabular}{ *{6}{c} }
\toprule
& $\bm{II}$ & $\bm{IZ}$ & $\bm{ZI}$ & $\bm{ZZ}$ & \textbf{others} \\
\midrule
$\bm{II}$ & $\invplus\frac{1}{4}$ & $\invplus\frac{1}{4}$ & $\invplus\frac{1}{4}$ & $-\frac{1}{4}$ & $0$ \\
$\bm{IZ}$ & $\invplus\frac{1}{4}$ & $\invplus\frac{1}{4}$ & $\invplus\frac{1}{4}$ & $-\frac{1}{4}$ & $0$ \\
$\bm{ZI}$ & $\invplus\frac{1}{4}$ & $\invplus\frac{1}{4}$ & $\invplus\frac{1}{4}$ & $-\frac{1}{4}$ & $0$ \\
$\bm{ZZ}$ & $-\frac{1}{4}$ & $-\frac{1}{4}$ & $-\frac{1}{4}$ & $\invplus\frac{1}{4}$ & $0$ \\
\textbf{others} & $\invplus 0$ & $\invplus 0$ & $\invplus 0$ & $\invplus 0$ & $0$ \\
\bottomrule
\end{tabular}
}
    \caption{Process matrix elements of the CZ gate under Pauli bases. The order of bases is as follows: $II,\allowbreak IX,\allowbreak IY,\allowbreak IZ,\allowbreak XI,\allowbreak XX,\allowbreak XY,\allowbreak XZ,\allowbreak YI,\allowbreak YX,\allowbreak YY,\allowbreak YZ,\allowbreak ZI,\allowbreak ZX,\allowbreak ZY,\allowbreak ZZ$.}
    \label{tab:CZPB}
\end{table}

The matrix form of the SWAP gate is represented as follows:
\begin{equation}
U_{\text{SWAP}}=\left[
\begin{array}{cccc}
    1&0&0&0\\
    0&0&1&0\\
    0&1&0&0\\
    0&0&0&1
\end{array}
\right]\!.
\end{equation}
Using Pauli bases representation, the SWAP gate can be written as ${U_{\text{SWAP}}=\frac{1}{2}(I\otimes I+X\otimes X+Y\otimes Y+Z\otimes Z)}$ of the tensor products of Pauli operators  $\{I,X,Y,Z\}$; see Table~\ref{tab:PBSWAP}.
\begin{table}[t]
\centering
\setlength{\extrarowheight}{4pt}
\scalebox{1.1}{
\begin{tabular}{ *{6}{c} }
\toprule
& $\bm{II}$ & $\bm{XX}$ & $\bm{YY}$ & $\bm{ZZ}$ & \textbf{others} \\
\midrule
$\bm{II}$ & $\invplus\frac{1}{4}$ & $\invplus\frac{1}{4}$ & $\invplus\frac{1}{4}$ & $\invplus\frac{1}{4}$ & $0$ \\
$\bm{XX}$ & $\invplus\frac{1}{4}$ & $\invplus\frac{1}{4}$ & $\invplus\frac{1}{4}$ & $\invplus\frac{1}{4}$ & $0$ \\
$\bm{YY}$ & $\invplus\frac{1}{4}$ & $\invplus\frac{1}{4}$ & $\invplus\frac{1}{4}$ & $\invplus\frac{1}{4}$ & $0$ \\
$\bm{ZZ}$ & $\invplus\frac{1}{4}$ & $\invplus\frac{1}{4}$ & $\invplus\frac{1}{4}$ & $\invplus\frac{1}{4}$ & $0$ \\
\textbf{others} & $\invplus 0$ & $\invplus 0$ & $\invplus 0$ & $\invplus 0$ & $0$ \\
\bottomrule
\end{tabular}
}
    \caption{Process matrix elements of the SWAP gate under Pauli bases. The order of bases is as follows: $II,\allowbreak IX,\allowbreak IY,\allowbreak IZ,\allowbreak XI,\allowbreak XX,\allowbreak XY,\allowbreak XZ,\allowbreak YI,\allowbreak YX,\allowbreak YY,\allowbreak YZ,\allowbreak ZI,\allowbreak ZX,\allowbreak ZY,\allowbreak ZZ$.}
    \label{tab:PBSWAP}
\end{table}

The reconstruction results for the CZ and SWAP gates are shown in Fig.~\ref{Fig_CZ_SWAP_supplement_tex}, while the results for the CNOT gate can be found in Fig.~\ref{Fig_CNOT_Gau_ave_100_N_1000} in the main text.
\begin{figure}[t]
    \includegraphics[width=.95\columnwidth]{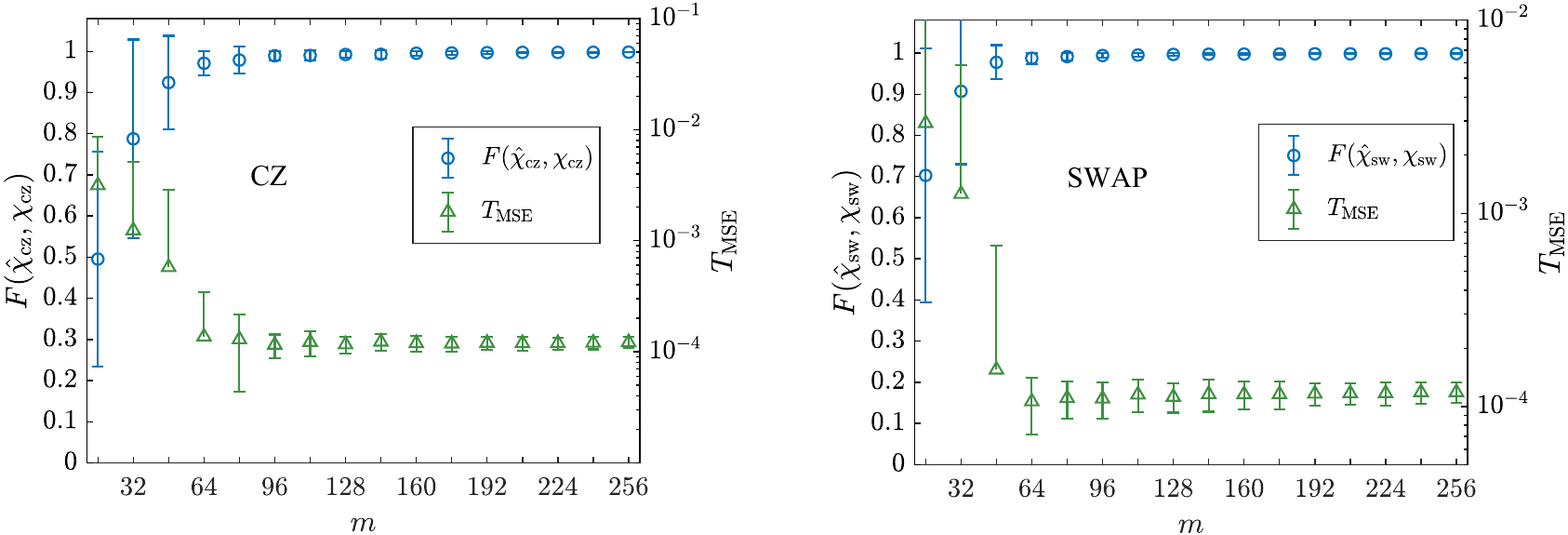}
    \caption{Fidelity $F(\hat{\chi},\chi)$ and MSE $T_{\text{MSE}}$ as functions of the number of configurations $m$  for the CZ and SWAP gates. The error bars are obtained from $100$ runs of tomography, with each run selecting $m$ combinations randomly. The blue points represent the fidelity between the reconstructed process matrix and true process matrix. The green points show the MSE between the reconstructed sparse Gaussian noise and true sparse Gaussian noise. The Gaussian noise has a mean of zero, a standard deviation of $1$, and a sparsity level of ${s = \lfloor 0.1m \rfloor}$. The
regularization parameters are chosen as ${\mu_1 = 10^{-5}, \mu_2 = 10^{-3}}$.}
    \label{Fig_CZ_SWAP_supplement_tex}
\end{figure}
\subsection{Three-qubit Toffoli and Fredkin gates}\label{app:3q_matrix}
The matrix form of three-qubit Toffoli gate is represented as follows:
\begin{equation}
U_{\text{Tof}}=\left[
\begin{array}{cccccccc}
    1&0&0&0&0&0&0&0\\
    0&1&0&0&0&0&0&0\\
    0&0&1&0&0&0&0&0\\
    0&0&0&1&0&0&0&0\\
    0&0&0&0&1&0&0&0\\
    0&0&0&0&0&1&0&0\\
    0&0&0&0&0&0&0&1\\
    0&0&0&0&0&0&1&0
\end{array}
\right]\!.
\end{equation}
Using Pauli bases representation, the Toffoli gate can be written as
\begin{equation}
\begin{aligned}
U_{\text{Tof}} &= \frac{1}{4} (3I \otimes I \otimes I + I \otimes I \otimes X + I \otimes Z \otimes I \\
&\quad - I \otimes Z \otimes X + Z \otimes I \otimes I - Z \otimes I \otimes X \\
&\quad -  Z \otimes Z \otimes I +  Z \otimes Z \otimes X)
\end{aligned}
\end{equation}
of the tensor products of Pauli operators  $\{I,X,Y,Z\}$; see Table~\ref{tab:TofPB}.
\begin{table}[t]
\centering
\setlength{\extrarowheight}{4pt}
\scalebox{1}{
\begin{tabular}{ *{10}{c} }
\toprule
& $\bm{III}$ & $\bm{IIX}$ & $\bm{IZI}$ & $\bm{IZX}$ & $\bm{ZII}$ & $\bm{ZIX}$ & $\bm{ZZI}$ & $\bm{ZZX}$ & \textbf{others} \\
\midrule
$\bm{III}$ & $\invplus\frac{9}{16}$ & $\invplus\frac{3}{16}$ & $\invplus\frac{3}{16}$ & $-\frac{3}{16}$ & $\invplus\frac{3}{16}$ & $-\frac{3}{16}$ & $-\frac{3}{16}$ & $\invplus\frac{3}{16}$ & $0$ \\
$\bm{IIX}$ & $\invplus\frac{3}{16}$ & $\invplus\frac{1}{16}$ & $\invplus\frac{1}{16}$ & $-\frac{1}{16}$ & $\invplus\frac{1}{16}$ & $-\frac{1}{16}$ & $-\frac{1}{16}$ & $\invplus\frac{1}{16}$ & $0$ \\
$\bm{IZI}$ & $\invplus\frac{3}{16}$ & $\invplus\frac{1}{16}$ & $\invplus\frac{1}{16}$ & $-\frac{1}{16}$ & $\invplus\frac{1}{16}$ & $-\frac{1}{16}$ & $-\frac{1}{16}$ & $\invplus\frac{1}{16}$ & $0$ \\
$\bm{IZX}$ & $-\frac{3}{16}$ & $-\frac{1}{16}$ & $-\frac{1}{16}$ & $\invplus\frac{1}{16}$ & $-\frac{1}{16}$ & $\invplus\frac{1}{16}$ & $\invplus\frac{1}{16}$ & $-\frac{1}{16}$ & $0$ \\
$\bm{ZII}$ & $\invplus\frac{3}{16}$ & $\invplus\frac{1}{16}$ & $\invplus\frac{1}{16}$ & $-\frac{1}{16}$ & $\invplus\frac{1}{16}$ & $-\frac{1}{16}$ & $-\frac{1}{16}$ & $\invplus\frac{1}{16}$ & $0$ \\
$\bm{ZIX}$ & $-\frac{3}{16}$ & $-\frac{1}{16}$ & $-\frac{1}{16}$ & $\invplus\frac{1}{16}$ & $-\frac{1}{16}$ & $\invplus\frac{1}{16}$ & $\invplus\frac{1}{16}$ & $-\frac{1}{16}$ & $0$ \\
$\bm{ZZI}$ & $-\frac{3}{16}$ & $-\frac{1}{16}$ & $-\frac{1}{16}$ & $\invplus\frac{1}{16}$ & $-\frac{1}{16}$ & $\invplus\frac{1}{16}$ & $\invplus\frac{1}{16}$ & $-\frac{1}{16}$ & $0$ \\
$\bm{ZZX}$ & $\invplus\frac{3}{16}$ & $\invplus\frac{1}{16}$ & $\invplus\frac{1}{16}$ & $-\frac{1}{16}$ & $\invplus\frac{1}{16}$ & $-\frac{1}{16}$ & $-\frac{1}{16}$ & $\invplus\frac{1}{16}$ & $0$ \\
\textbf{others} & $\invplus 0$ & $\invplus 0$ & $\invplus 0$ & $\invplus 0$ & $\invplus 0$ & $\invplus 0$ & $\invplus 0$ & $\invplus 0$ & $0$ \\
\bottomrule
\end{tabular}
}
    \caption{Process matrix elements of the Toffoli gate under Pauli bases. The order of the Pauli bases is as follows: $III,\allowbreak IIX,\allowbreak IIY,\allowbreak IIZ,...,\allowbreak ZZZ$.}
    \label{tab:TofPB}
\end{table}

The matrix form of three-qubit Fredkin gate is given as
\begin{equation}
U_{\text{Fred}}=\left[
\begin{array}{cccccccc}
    1&0&0&0&0&0&0&0\\
    0&1&0&0&0&0&0&0\\
    0&0&1&0&0&0&0&0\\
    0&0&0&1&0&0&0&0\\
    0&0&0&0&1&0&0&0\\
    0&0&0&0&0&0&1&0\\
    0&0&0&0&0&1&0&0\\
    0&0&0&0&0&0&0&1
\end{array}
\right]\!.
\end{equation}
Using Pauli bases representation, the Fredkin gate can be written as
\begin{equation}
\begin{aligned}
U_{\text{Fred}} &= \frac{1}{4} (3I \otimes I \otimes I + I \otimes X \otimes X + I \otimes Y \otimes Y \\
&\quad + I \otimes Z \otimes Z + Z \otimes I \otimes I - Z \otimes X \otimes X \\
&\quad -  Z \otimes Y \otimes Y -  Z \otimes Z \otimes Z)
\end{aligned}
\end{equation}
of the tensor products of Pauli operators $\{I,X,Y,Z\}$; see Table~\ref{tab:FredPB}.
\begin{table}[t]
\centering
\setlength{\extrarowheight}{4pt}
\scalebox{1}{
\begin{tabular}{ *{10}{c} }
\toprule
& $\bm{III}$ & $\bm{IXX}$ & $\bm{IYY}$ & $\bm{IZZ}$ & $\bm{ZII}$ & $\bm{ZXX}$ & $\bm{ZYY}$ & $\bm{ZZZ}$ & \textbf{others} \\
\midrule
$\bm{III}$ & $\invplus\frac{9}{16}$ & $\invplus\frac{3}{16}$ & $\invplus\frac{3}{16}$ & $\invplus\frac{3}{16}$ & $\invplus\frac{3}{16}$ & $-\frac{3}{16}$ & $-\frac{3}{16}$ & $-\frac{3}{16}$ & $0$ \\
$\bm{IXX}$ & $\invplus\frac{3}{16}$ & $\invplus\frac{1}{16}$ & $\invplus\frac{1}{16}$ & $\invplus\frac{1}{16}$ & $\invplus\frac{1}{16}$ & $-\frac{1}{16}$ & $-\frac{1}{16}$ & $-\frac{1}{16}$ & $0$ \\
$\bm{IYY}$ & $\invplus\frac{3}{16}$ & $\invplus\frac{1}{16}$ & $\invplus\frac{1}{16}$ & $\invplus\frac{1}{16}$ & $\invplus\frac{1}{16}$ & $-\frac{1}{16}$ & $-\frac{1}{16}$ & $-\frac{1}{16}$ & $0$ \\
$\bm{IZZ}$ & $\invplus\frac{3}{16}$ & $\invplus\frac{1}{16}$ & $\invplus\frac{1}{16}$ & $\invplus\frac{1}{16}$ & $\invplus\frac{1}{16}$ & $-\frac{1}{16}$ & $-\frac{1}{16}$ & $-\frac{1}{16}$ & $0$ \\
$\bm{ZII}$ & $\invplus\frac{3}{16}$ & $\invplus\frac{1}{16}$ & $\invplus\frac{1}{16}$ & $\invplus\frac{1}{16}$ & $\invplus\frac{1}{16}$ & $-\frac{1}{16}$ & $-\frac{1}{16}$ & $-\frac{1}{16}$ & $0$ \\
$\bm{ZXX}$ & $-\frac{3}{16}$ & $-\frac{1}{16}$ & $-\frac{1}{16}$ & $-\frac{1}{16}$ & $-\frac{1}{16}$ & $\invplus\frac{1}{16}$ & $\invplus\frac{1}{16}$ & $\invplus\frac{1}{16}$ & $0$ \\
$\bm{ZYY}$ & $-\frac{3}{16}$ & $-\frac{1}{16}$ & $-\frac{1}{16}$ & $-\frac{1}{16}$ & $-\frac{1}{16}$ & $\invplus\frac{1}{16}$ & $\invplus\frac{1}{16}$ & $\invplus\frac{1}{16}$ & $0$ \\
$\bm{ZZZ}$ & $-\frac{3}{16}$ & $-\frac{1}{16}$ & $-\frac{1}{16}$ & $-\frac{1}{16}$ & $-\frac{1}{16}$ & $\invplus\frac{1}{16}$ & $\invplus\frac{1}{16}$ & $\invplus\frac{1}{16}$ & $0$ \\
\textbf{others} & $\invplus 0$ & $\invplus 0$ & $\invplus 0$ & $\invplus 0$ & $\invplus 0$ & $\invplus 0$ & $\invplus 0$ & $\invplus 0$ & $0$ \\
\bottomrule
\end{tabular}
}
    \caption{Process matrix elements of the Fredkin gate under Pauli bases. The order of the Pauli bases is as follows: $III,\allowbreak IIX,\allowbreak IIY,\allowbreak IIZ,...,\allowbreak ZZZ$.}
    \label{tab:FredPB}
\end{table}
The reconstruction results for the Fredkin gate are shown in Fig.~\ref{Fig_Fred_Gau_ave_50_N_1000} as a supplementary illustration to the main text.
\begin{figure}[t]
    \includegraphics[width=.5\columnwidth]{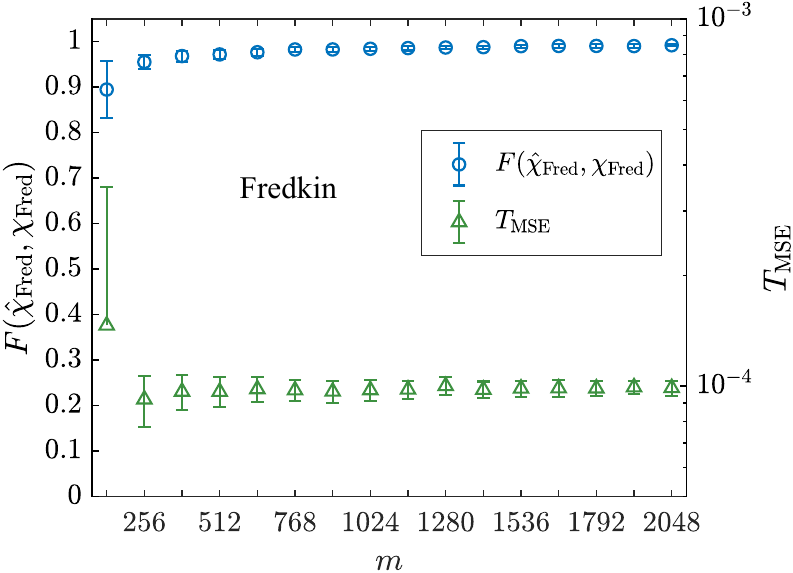}
    \caption{Fidelity $F(\hat{\chi}_{\text{Fred}},\chi_{\text{Fred}})$ and MSE $T_{\text{MSE}}$ as functions of the number of configurations $m$. The error bars are obtained from $50$ runs of tomography, with each run selecting $m$ combinations randomly. The blue points represent the fidelity between the reconstructed process matrix and true process matrix of the Fredkin gate. The green points show the MSE between the reconstructed sparse Gaussian noise and true sparse Gaussian noise. The Gaussian noise has a mean of zero, a standard deviation of $1$, and a sparsity level of ${s = \lfloor 0.1m \rfloor}$. The
regularization parameters are chosen as ${\mu_1 = 10^{-5}, \mu_2 = 10^{-3}}$.}
    \label{Fig_Fred_Gau_ave_50_N_1000}
\end{figure}

\subsection{Four-qubit CCCZ gate}\label{app:CCCZ}
Using Pauli bases representation, the CCCZ gate can be written as
\begin{equation}
\begin{aligned}
U_{\text{CCCZ}} &=\frac{1}{8}(7IIII+IIIZ+IIZI-IIZZ\\
&+IZII-IZIZ-IZZI+IZZZ\\
&+ZIII-ZIIZ-ZIZI+ZIZZ\\
&-ZZII+ZZIZ+ZZZI-ZZZZ)\,.
\end{aligned}
\end{equation}
See Table~\ref{tab:CCCZPB} for the  process matrix elements of the CCCZ gate under Pauli bases.
\begin{table}[t]
\centering
\setlength{\extrarowheight}{4pt} 
\setlength{\tabcolsep}{1pt}
\scalebox{1}{ 
\begin{tabular}{ *{18}{c} } 
\toprule
& $\bm{IIII}$ & $\bm{IIIZ}$ & $\bm{IIZI}$ & $\bm{IIZZ}$ & $\bm{IZII}$ & $\bm{IZIZ}$ & $\bm{IZZI}$ & $\bm{IZZZ}$ &$\bm{ZIII}$ & $\bm{ZIIZ}$ & $\bm{ZIZI}$ & $\bm{ZIZZ}$ & $\bm{ZZII}$ & $\bm{ZZIZ}$ & $\bm{ZZZI}$ & $\bm{ZZZZ}$ & \textbf{others} \\
\midrule
$\bm{IIII}$ & $\invplus \frac{49}{64}$ & $\invplus \frac{7}{64}$ & $\invplus \frac{7}{64}$ & $-\frac{7}{64}$ & $\invplus \frac{7}{64}$ & $-\frac{7}{64}$ & $-\frac{7}{64}$ & $\invplus \frac{7}{64}$ &
$\invplus \frac{7}{64}$ & $-\frac{7}{64}$ & $-\frac{7}{64}$ & $\invplus \frac{7}{64}$ & $-\frac{7}{64}$ & $\invplus \frac{7}{64}$ & $\invplus \frac{7}{64}$ & $-\frac{7}{64}$ & $0$ \\
$\bm{IIIZ}$ & $\invplus \frac{7}{64}$ & $\invplus \frac{1}{64}$ & $\invplus \frac{1}{64}$ & $-\frac{1}{64}$ & $\invplus \frac{1}{64}$ & $-\frac{1}{64}$ & $-\frac{1}{64}$ & $\invplus \frac{1}{64}$ &
$\invplus \frac{1}{64}$ & $-\frac{1}{64}$ & $-\frac{1}{64}$ & $\invplus \frac{1}{64}$ & $-\frac{1}{64}$ & $\invplus \frac{1}{64}$ & $\invplus \frac{1}{64}$ & $-\frac{1}{64}$ & $0$\\
$\bm{IIZI}$ & $\invplus \frac{7}{64}$ & $\invplus \frac{1}{64}$ & $\invplus \frac{1}{64}$ & $-\frac{1}{64}$ & $\invplus \frac{1}{64}$ & $-\frac{1}{64}$ & $-\frac{1}{64}$ & $\invplus \frac{1}{64}$ &
$\invplus \frac{1}{64}$ & $-\frac{1}{64}$ & $-\frac{1}{64}$ & $\invplus \frac{1}{64}$ & $-\frac{1}{64}$ & $\invplus \frac{1}{64}$ & $\invplus \frac{1}{64}$ & $-\frac{1}{64}$ & $0$ \\
$\bm{IIZZ}$ & $-\frac{7}{64}$ & $-\frac{1}{64}$ & $-\frac{1}{64}$ & $\invplus \frac{1}{64}$ & $-\frac{1}{64}$ & $\invplus \frac{1}{64}$ & $\invplus \frac{1}{64}$ & $-\frac{1}{64}$ &
$-\frac{1}{64}$ & $\invplus \frac{1}{64}$ & $\invplus \frac{1}{64}$ & $-\frac{1}{64}$ & $\invplus \frac{1}{64}$ & $-\frac{1}{64}$ & $-\frac{1}{64}$ & $\invplus \frac{1}{64}$ & $0$ \\
$\bm{IZII}$ & $\invplus \frac{7}{64}$ & $\invplus \frac{1}{64}$ & $\invplus \frac{1}{64}$ & $-\frac{1}{64}$ & $\invplus \frac{1}{64}$ & $-\frac{1}{64}$ & $-\frac{1}{64}$ & $\invplus \frac{1}{64}$ &
$\invplus \frac{1}{64}$ & $-\frac{1}{64}$ & $-\frac{1}{64}$ & $\invplus \frac{1}{64}$ & $-\frac{1}{64}$ & $\invplus \frac{1}{64}$ & $\invplus \frac{1}{64}$ & $-\frac{1}{64}$ & $0$ \\
$\bm{IZIZ}$ & $-\frac{7}{64}$ & $-\frac{1}{64}$ & $-\frac{1}{64}$ & $\invplus \frac{1}{64}$ & $-\frac{1}{64}$ & $\invplus \frac{1}{64}$ & $\invplus \frac{1}{64}$ & $-\frac{1}{64}$ &
$-\frac{1}{64}$ & $\invplus \frac{1}{64}$ & $\invplus \frac{1}{64}$ & $-\frac{1}{64}$ & $\invplus \frac{1}{64}$ & $-\frac{1}{64}$ & $-\frac{1}{64}$ & $\invplus \frac{1}{64}$ & $0$ \\
$\bm{IZZI}$ & $-\frac{7}{64}$ & $-\frac{1}{64}$ & $-\frac{1}{64}$ & $\invplus \frac{1}{64}$ & $-\frac{1}{64}$ & $\invplus \frac{1}{64}$ & $\invplus \frac{1}{64}$ & $-\frac{1}{64}$ &
$-\frac{1}{64}$ & $\invplus \frac{1}{64}$ & $\invplus \frac{1}{64}$ & $-\frac{1}{64}$ & $\invplus \frac{1}{64}$ & $-\frac{1}{64}$ & $-\frac{1}{64}$ & $\invplus \frac{1}{64}$ & $0$ \\
$\bm{IZZZ}$ & $\invplus \frac{7}{64}$ & $\invplus \frac{1}{64}$ & $\invplus \frac{1}{64}$ & $-\frac{1}{64}$ & $\invplus \frac{1}{64}$ & $-\frac{1}{64}$ & $-\frac{1}{64}$ & $\invplus \frac{1}{64}$ &
$\invplus \frac{1}{64}$ & $-\frac{1}{64}$ & $-\frac{1}{64}$ & $\invplus \frac{1}{64}$ & $-\frac{1}{64}$ & $\invplus \frac{1}{64}$ & $\invplus \frac{1}{64}$ & $-\frac{1}{64}$ & $0$ \\
$\bm{ZIII}$ & $\invplus \frac{7}{64}$ & $\invplus \frac{1}{64}$ & $\invplus \frac{1}{64}$ & $-\frac{1}{64}$ & $\invplus \frac{1}{64}$ & $-\frac{1}{64}$ & $-\frac{1}{64}$ & $\invplus \frac{1}{64}$ &
$\invplus \frac{1}{64}$ & $-\frac{1}{64}$ & $-\frac{1}{64}$ & $\invplus \frac{1}{64}$ & $-\frac{1}{64}$ & $\invplus \frac{1}{64}$ & $\invplus \frac{1}{64}$ & $-\frac{1}{64}$ & $0$ \\
$\bm{ZIIZ}$ & $-\frac{7}{64}$ & $-\frac{1}{64}$ & $-\frac{1}{64}$ & $\invplus \frac{1}{64}$ & $-\frac{1}{64}$ & $\invplus \frac{1}{64}$ & $\invplus \frac{1}{64}$ & $-\frac{1}{64}$ &
$-\frac{1}{64}$ & $\invplus \frac{1}{64}$ & $\invplus \frac{1}{64}$ & $-\frac{1}{64}$ & $\invplus \frac{1}{64}$ & $-\frac{1}{64}$ & $-\frac{1}{64}$ & $\invplus \frac{1}{64}$ & $0$ \\
$\bm{ZIZI}$ & $-\frac{7}{64}$ & $-\frac{1}{64}$ & $-\frac{1}{64}$ & $\invplus \frac{1}{64}$ & $-\frac{1}{64}$ & $\invplus \frac{1}{64}$ & $\invplus \frac{1}{64}$ & $-\frac{1}{64}$ &
$-\frac{1}{64}$ & $\invplus \frac{1}{64}$ & $\invplus \frac{1}{64}$ & $-\frac{1}{64}$ & $\invplus \frac{1}{64}$ & $-\frac{1}{64}$ & $-\frac{1}{64}$ & $\invplus \frac{1}{64}$ & $0$ \\
$\bm{ZIZZ}$ & $\invplus \frac{7}{64}$ & $\invplus \frac{1}{64}$ & $\invplus \frac{1}{64}$ & $-\frac{1}{64}$ & $\invplus \frac{1}{64}$ & $-\frac{1}{64}$ & $-\frac{1}{64}$ & $\invplus \frac{1}{64}$ &
$\invplus \frac{1}{64}$ & $-\frac{1}{64}$ & $-\frac{1}{64}$ & $\invplus \frac{1}{64}$ & $-\frac{1}{64}$ & $\invplus \frac{1}{64}$ & $\invplus \frac{1}{64}$ & $-\frac{1}{64}$ & $0$ \\
$\bm{ZZII}$ & $-\frac{7}{64}$ & $-\frac{1}{64}$ & $-\frac{1}{64}$ & $\invplus \frac{1}{64}$ & $-\frac{1}{64}$ & $\invplus \frac{1}{64}$ & $\invplus \frac{1}{64}$ & $-\frac{1}{64}$ &
$-\frac{1}{64}$ & $\invplus \frac{1}{64}$ & $\invplus \frac{1}{64}$ & $-\frac{1}{64}$ & $\invplus \frac{1}{64}$ & $-\frac{1}{64}$ & $-\frac{1}{64}$ & $\invplus \frac{1}{64}$ & $0$ \\
$\bm{ZZIZ}$ & $\invplus \frac{7}{64}$ & $\invplus \frac{1}{64}$ & $\invplus \frac{1}{64}$ & $-\frac{1}{64}$ & $\invplus \frac{1}{64}$ & $-\frac{1}{64}$ & $-\frac{1}{64}$ & $\invplus \frac{1}{64}$ &
$\invplus \frac{1}{64}$ & $-\frac{1}{64}$ & $-\frac{1}{64}$ & $\invplus \frac{1}{64}$ & $-\frac{1}{64}$ & $\invplus \frac{1}{64}$ & $\invplus \frac{1}{64}$ & $-\frac{1}{64}$ & $0$\\
$\bm{ZZZI}$ & $\invplus \frac{7}{64}$ & $\invplus \frac{1}{64}$ & $\invplus \frac{1}{64}$ & $-\frac{1}{64}$ & $\invplus \frac{1}{64}$ & $-\frac{1}{64}$ & $-\frac{1}{64}$ & $\invplus \frac{1}{64}$ &
$\invplus \frac{1}{64}$ & $-\frac{1}{64}$ & $-\frac{1}{64}$ & $\invplus \frac{1}{64}$ & $-\frac{1}{64}$ & $\invplus \frac{1}{64}$ & $\invplus \frac{1}{64}$ & $-\frac{1}{64}$ & $0$ \\
$\bm{ZZZZ}$ & $-\frac{7}{64}$ & $-\frac{1}{64}$ & $-\frac{1}{64}$ & $\invplus \frac{1}{64}$ & $-\frac{1}{64}$ & $\invplus \frac{1}{64}$ & $\invplus \frac{1}{64}$ & $-\frac{1}{64}$ &
$-\frac{1}{64}$ & $\invplus\frac{1}{64}$ & $\invplus\frac{1}{64}$ & $-\frac{1}{64}$ & $\invplus\frac{1}{64}$ & $-\frac{1}{64}$ & $-\frac{1}{64}$ & $\invplus\frac{1}{64}$ & $0$ \\
\textbf{others} & $\invplus 0$ & $\invplus 0$ & $\invplus 0$ & $\invplus 0$ & $\invplus 0$ & $\invplus 0$ & $\invplus 0$ & $\invplus 0$ & $\invplus 0$ & $\invplus 0$ & $\invplus 0$ & $\invplus 0$ & $\invplus 0$ & $\invplus 0$ & $\invplus 0$ & $\invplus 0$ & $0$ \\
\bottomrule 
\end{tabular} }
    \caption{Process matrix elements of the CCCZ gate under Pauli bases. The order of the Pauli bases is as follows: $IIII,\allowbreak IIIX,\allowbreak IIIY,\allowbreak IIIZ,\allowbreak ...,\allowbreak ZZZZ$.}
    \label{tab:CCCZPB}
\end{table}

\end{document}